\documentclass[a4,11pt]{article}
\newtheorem{lemma}{Lemma}
\newtheorem{theorem}{Theorem}%[section]
\newenvironment{proof}%
{\begin{trivlist}\item[\hspace*{\labelsep}{\it Proof.\/}]}%
{\hfill$\Box$\end{trivlist}}
\newtheorem{theo}{Theorem}

\newtheorem{defi}[theo]{Definition}

\newcommand{\para}{\medskip\noindent}

\newcommand{\head}[1]
 {\markright{\hbox to 0pt{\vtop to 0pt{\hbox{}\vskip 3mm \hrule
 width  \textwidth \vss} \hss}{\sc #1}}}
\usepackage{latexsym,graphicx,amsfonts,amsmath}
\begin{document}

%%\begin{titlepage}
\title{\bf A New Upper Bound on 2D Online Bin Packing} %\thanks{Supported in part by }}

\author{Xin Han$^1$, Francis Y.L. Chin$^2$, Hing-Fung Ting$^2$, Guochuan Zhang$^3$, Yong Zhang$^2$
% %%\hspace{3mm} Kazuo Iwama$^1$\hspace{3mm} Guochuan
% %%Zhang$^{2}$
% %%$^3$\thanks{Supported in part by the DFG Project AL
% %%464/4-1, Eu-Project APPOL II and NSFC (10231060).}
% \\ \\ {\small $^1$ School of Informatics, Kyoto University, Kyoto
% 606-8501, Japan} \\ {\small \{hanxin,
% iwama\}@kuis.kyoto-u.ac.jp}
% %%\\ {\small $^2$ Institut f\"ur Informatik, Univerist\"at Freiburg}
% %%\\ {\small Georges-K\"ohler-Allee 79, 79110 Freiburg, Germany}
\\{\small $^1$ 
Graduate School of Information and Technology,
University of Tokyo}
%%Tokyo, 113-8656, Japan}
\\{\small hanxin.mail@gmail.com}
\\ {\small $^2$ Department of Computer Science,
The University of Hong Kong, Hong Kong}
 \\{\small \{ chin, hfting, yzhang\}@cs.hku.hk}
 \\ {\small $^3$ Department of Mathematics, Zhejiang University, China}
 \\ {\small zgc@zju.edu.cn}}
%%%\date{}
% and Department of Mathematics, Zhejiang
%University, Hangzhou 310027, China. Research }}
%\pagestyle{heardings}
%\markright{Kyoto University, China, HAN Xin, Online removable square packing}
\maketitle

%\lhead{Kyoto University, China, HAN Xin, Online removable square packing}
\baselineskip 14pt

\begin{abstract}
%  In this paper, we study two dimensional online  bin packing problem
%   (for short, 2D BP). Seiden and van Stee proposed a nice algorithm
%   called $A \otimes B$ for this problem and proved that the algorithm
%   has its performace ratio at most 2.66013, where $A$ and $B$ are two
%   instances of Improved Harmonic algorithm.
%   It is natural  to   think about that if we  let $A$ and $B$
%   be instances of Super Harmonic(which was obtained in 2001 and is more
%   powerful than Improved Harmonic algorithm),
%   then the previous upper bound would be improved.
%   However, as mentioned in paper \cite{SS03}, in that case,
%   the previous analysis  framework does  not go through.
%   The reason is that the weighting functions for Super Harmonic defined
%   in \cite{S02} do not work well with the way to
%   bound the total weight in a single bin in paper \cite{SS03}.

The 2D Online Bin Packing is a fundamental problem in Computer Science
and the determination of its asymptotic competitive ratio has
attracted great research attention.   In a long series of papers, the
lower bound of this ratio has been improved from 1.808, 1.856 to 1.907
 and its upper bound reduced from 3.25, 3.0625, 2.8596, 2.7834 to 2.66013. 
 In this paper, we rewrite the upper bound record to 2.5545.
Our idea for the improvement is as follows. 

In SODA 2002 \cite{SS03}, Seiden and van Stee proposed an elegant
algorithm called $H \otimes B$, comprised of the {\em Harmonic algorithm} 
$H$ and the {\em Improved Harmonic algorithm} $B$, 
for the two-dimensional online bin packing problem and proved that 
the algorithm has an asymptotic competitive ratio of at most 2.66013. Since
the best known online algorithm for one-dimensional bin packing is
the {\em Super Harmonic algorithm} \cite{S02}, a natural question to
ask is: could a better upper bound be achieved by using the 
Super Harmonic algorithm  instead of the Improved Harmonic
algorithm?  However, as mentioned in~\cite{SS03}, the previous
analysis framework does not work. In this paper, we give a positive
answer for the above question. A new upper bound of 2.5545 is obtained
for 2-dimensional online bin packing. The main idea is to develop
new weighting functions for the  Super Harmonic algorithm and
propose new techniques to bound the total weight in a rectangular
bin.
% iii) finally, we show that the proposed weighting functions work very well
% with the new  approach of bounding the total weight in a single bin,
% and a new upper bound .%

% A byproduct of this paper is that we propose
%new weighting functions for Super Harmonic which is
%currently the best performing online algorithm for 1D BP.
 \end{abstract}

%%\end{titlepage}
%%\newpage
 \baselineskip 12.9pt
\section{Introduction}
In two-dimensional bin packing, each item $(w_i,h_i)$ is a rectangle
of width $w_i\le 1$ and height $h_i\le 1$. Given a list of such
rectangular items, one is asked to pack all of them into a minimum
number of square bins of side length one so that their sides are
parallel to the sides of the bin. Rotation is not allowed. The
problem is clearly strongly NP-hard since it is a generalization of
the one-dimensional bin packing problem \cite{CGJ97}. In this paper
we will consider the online version of two-dimensional bin packing,
in which the items are released one by one and we must irrevocably
pack the current item into a bin without any information on the next
items. Before presenting the previous results and our work, we first
review the standard measure for online bin packing algorithms.

\paragraph {\bf Asymptotic competitive ratio}
To evaluate an online algorithm for bin packing problems, we use the
{\em asymptotic competitive ratio} defined as follows. Consider an
online algorithm $A$. For any list $L$ of items, let $A(L)$ be the
cost (number of bins used) incurred by algorithm $A$ and let $OPT(L)$ be
the corresponding optimal value. Then the {\em asymptotic
competitive ratio} for algorithm $A$ is
\[
     R_A^{\infty} =\lim_{k \to \infty} \sup \max_{L}\{ A(L)/OPT(L)| OPT(L) = k\}.
\]

\paragraph{\bf Previous work}
Bin packing has been well-studied. For the
one-dimensional case, Johnson et al.~\cite{JDUGG74} showed that the
First Fit algorithm (FF) has an asymptotic competitive ratio of 1.7.
Yao~\cite{Yao80} improved algorithm FF with a better upper bound of
5/3. Lee and Lee~\cite{LL85} introduced the class of Harmonic
algorithms, for which an asymptotic competitive
ratio of 1.63597 was achieved. Ramanan et al.~\cite{RBLL89} further improved the
upper bound to 1.61217. The best known upper bound so far is from
the Super Harmonic algorithm by Seiden~\cite{S02} whose asymptotic
competitive ratio is at most 1.58889. As for the negative results,
Yao~\cite{Yao80} showed that no online algorithm has asymptotic
competitive ratio less than 1.5. Brown~\cite{B79} and
Liang~\cite{L80} independently provided a better lower bound of
1.53635. The best known lower bound to date is 1.54014~\cite{Vliet92}.
%For online results, the first APTAS was given in \cite{FL81} and an
%AFPTAS was given by Karmarkar and Karp \cite{KK82}.

As for two-dimensional online bin packing, a lower bound of 1.6 was
given by Galambos \cite{Gal91}. The result was gradually improved to
1.808 \cite{GvV94}, 1.857 \cite{Vliet95} and 1.907 \cite{BvVW96}.
Coppersmith and Raghan \cite{CR89} gave the first online algorithm
with asymptotic competitive ratio 3.25. Csirik et al. \cite{CFL93}
improved the upper bound to 3.0625. Csirik and van Vliet \cite{CV93}
presented an algorithm for all $d$ dimensions, where in particular for two
dimensions, they obtained a ratio of at most 2.8596.
Based on the techniques on the Improved Harmonic, 
Han et.al \cite{HFG01} improved the upper bound to 2.7834. 
 The best known online
algorithm to date is the one called $A \otimes B$ presented by
Seiden and van Stee \cite{SS03}, where $A$ and $B$ stand for two
one-dimensional online bin packing algorithms. Basically $A$ and $B$
are applied to one dimension of the items with rounding sizes. In
this seminal paper Seiden and van Stee proved that the asymptotic
competitive ratio of $H \otimes B$ is at most 2.66013, where $H$ is
the Harmonic algorithm \cite{LL85} and $B$ is an instance of the
improved Harmonic algorithm. It has been open since then to improve
the upper bound. A natural idea is to use an instance of the Super
Harmonic algorithm \cite{S02} instead of the improved Harmonic algorithm. 
However, as mentioned in paper \cite{SS03}, in that case, the previous analysis
framework  cannot be extended to Super Harmonic.

We also briefly overview the offline results on two-dimensional bin
packing. Chung et al \cite{CGJ82} showed an approximation algorithm
with an asymptotic performance ratio of 2.125. Caprara \cite{Cap02}
improved the upper bound to 1.69103. Very recently Bansal et al.
\cite{BCS06} derived a randomized algorithm with asymptotic
performance ratio of at most 1.525. As for the negative results, Bansal
et al. \cite{BCKS06} showed that the two-dimensional bin packing
problem does not admit an asymptotic polynomial time approximation
scheme.

For the special case where items are squares, there is also a large
number of results \cite{CR89, SS03, MW03, ES04,ES05a,ES05b, HYZ06}.
Especially for bounded space online algorithms, Epstein et al.
\cite{ES05a} gave an optimal online algorithm.

\para{\bf Our contributions}
There are two main contributions in this paper,
 \begin{itemize} 
\item we revisit 1D online bin packing algorithm: Super Harmonic, 
   give new weighting functions for it,  which are much simpler than
   the ones introduced in \cite{S02}, 
   and the new weighting functions have interests in its own. 
\item we generalize the previous analysis framework  for 2D online bin 
    packing algorithms used in \cite{SS03}, 
    and show that the new analysis framework are  very useful 
    in  analyzing 2D or multi-dimensional  online bin packing problems. 
\end{itemize}
By combining the new weighting functions with the new analysis framework,
we design  a new 2D online bin packing algorithm with 
a competitive ratio  2.5545, 
which  improves the previous bound of 2.66013 in SODA 2002 \cite{SS03}.
As mentioned in \cite{SS03},
the old analysis framework does not work well with the old weighting functions
in \cite{S02}, i.e., the old approach does not guarantee an upper bound
better than 2.66013. 
This is testified in the following way:
consider our  algorithm, 
if we  use {\em old} weighting functions with the {\em old}  framework to
analyze it, the competitive ratio is at least 3.04,
and if we  use the {\em old} weighting functions with 
the {\em new} framework, the competitive ratio is at least 2.79.  
%% So, we can say that the new weighting function works well with
%% the new framework.
%% Our goal is to overcome the difficulties in the previous proof
%% system \cite{SS03} and show that $H \otimes B$ does have a better
%% upper bound if $B$ is an instance of the Super Harmonic Algorithm.
%% The crucial point is defining the weight of each item so that the
%% weight of a bin is perfectly bounded. To this end, we develop new
%% weighting functions for the Super Harmonic algorithm and propose new
%% techniques to bound the weight of a bin. Eventually we show that the
%% proposed weighting functions work very well and a new upper bound of
%% 2.5545 is thus obtained, which improves the previous bound of 2.66013
%% in SODA 2002 \cite{SS03}.

\para{\bf Organization of Paper}
Section 2 will review the Harmonic and Super Harmonic algorithms
as preliminaries. Section 3 defines the weighting functions for Super
Harmonic. Section 4 describes and analyzes the two-dimensional online 
bin packing algorithm $H \otimes SH+$. Section 5 concludes.

\section{Preliminaries}
We first review two online algorithms for one-dimensional bin
packing, Harmonic and Super Harmonic, which are employed in
designing online algorithms for two-dimensional bin packing.

\subsection{The Harmonic algorithm}
The Harmonic algorithm is a fundamental bin packing algorithm with a
simple and nice structure, that was introduced by Lee and Lee
\cite{LL85} in 1985. The algorithm works as follows. Given a positive
integer $k$, each item is immediately classified into one of $k$ types 
according to its size upon its arrival.
In particular, if an item has a size in interval $(\frac{1}{i+1}, \frac{1}{i}]$ for
some integer $i$, where $1 \le i <k$, then it is a type-$i$ item; otherwise, it is of type-$k$. 
The type-$i$ item is then packed, using the simple Next Fit (NF) algorithm, into
the open (not fully-packed) bin designated to contain type-$i$ items exclusively;
new bins are opened when necessary. At any time, there is at most one 
open bin for each type and any closed (fully-packed) bin for type-$i$ is packed
exactly with $i$ items of type-$i$ for $1\le i<k$.

For an item of size $x$, we define a weighting function $W_H(x)$ for
the Harmonic algorithm as follows:
  \begin{displaymath}
   W_H(x) = \left\{ \begin{array}{ll}
                     \frac1i, & \textrm{ if $\frac1{i+1} < x \le \frac1i$ with
       $1 \le i < k$,}  \\
                    \frac{k}{k-1}x & \textrm{ if $0 < x \le \frac1k$.}
                    \end{array}
            \right.
  \end{displaymath}
The following lemma is directly from \cite{LL85}.
\begin{lemma}
For any list $L$, we have
\[
 H(L) \le \sum_{p \in L}W_{H}(p) + O(1),
\]
 where $H(L)$ is the number of bins used by the Harmonic algorithm for list $L$.
\end{lemma}

\subsection{The Super Harmonic algorithm}
The Super Harmonic algorithm \cite{S02} is a generalization of the
Improved Harmonic algorithm and the Harmonic algorithm. Super
Harmonic first classifies each item into one of $k+1$ types, where $k$ is a
positive integer, and then assigns to the item a color of either {\em blue}
or {\em red}. It allows items of up to two different types to share the same
bin. In any one bin, 
all items of the same type have same color and items
of different type have different colors. For items of  type-$i$ ($i\le k$), 
the algorithm maintains two
parameters $\beta_i$ and $\gamma_i$ to bound respectively the number of blue
items and the number of red items in a bin. More details are given
below.

\paragraph {\bf Classification into types} Let
$t_1 = 1 > t_2 >...>$ $t_{k} > t_{k+1} = \epsilon > t_{k+2}=0 $ be
real numbers. An interval $I_i$ is defined to be $(t_{i+1}, t_i]$,
for $i= 1,...,k+1$. An item with size $x$ is of type-$i$ if $x \in I_i$.

\paragraph {\bf Coloring red or blue}
Each type-$i$ item is also assigned a color, either red or
blue, for $i \le k$. The algorithm uses two sets of counters,
$e_1,...,e_k$ and $s_1,...,s_k$, all of which are initially zero.
The total number of type-$i$ items is denoted by $s_i$, while the
number of type-$i$ red items is denoted by $e_i$. For $1 \le i
\le k$, during the packing process,  the fraction
of type-$i$ items that are red is maintained, i.e., $e_i=\lfloor \alpha_i
s_i \rfloor$, where $\alpha_1,...,\alpha_k \in$ [0,1] are constants.

\paragraph {\bf Maximal number of blue items}
Let $\beta_i = \lfloor \frac{1}{t_{i}} \rfloor$ for $1 \le i \le k$,
which is the maximal number of blue items of type-$i$ which can be
accepted in a single bin. 

\paragraph {\bf Space left for red items}
Let $\delta_i = 1 -t_i\beta_i$, which is
the lower bound of the space left when a bin consists of $\beta_i$
blue items of type-$i$. If possible, we want to use the space left for
small red items. Note that in the algorithm, in order to simplify the 
analysis, instead of using
$\delta_i$, less space is used, namely $D=\{\Delta_0, \Delta_1,...,\Delta_K\}$,
as the spaces into which red items can be packed, where $0 =
\Delta_0 < \Delta_1 < \cdots < \Delta_{K} < 1/2$ and $K \le k$.
Let $\Delta_{\phi(i)}$ be the space to be used to accommodate
red items in a bin which holds $\beta_i$ blue items of
type-$i$, where function $\phi$ is defined as \{1,...,k\} $\mapsto$
\{0,...,K\} such that $\phi$ satisfies $\Delta_{\phi(i)} \le \delta_{i}$.
If $\phi(i) = 0$ then no red items are accepted. 

For convenient use in our analysis in the next section,
we introduce a function called $\varphi(i)$, which gives the index of the 
smallest space in $D$ into which a red item of type-$i$ can be placed:

\[
  \varphi(i) = \min\{j |  t_i \le \Delta_j, 1 \le j \le K \}.
\]

\paragraph {\bf Maximal number of red items}
Now we define $\gamma_i$. Let $\gamma_i=0$ if $t_i
> \Delta_{K}$; otherwise $\gamma_{i}= max\{1,\lfloor \Delta_1/t_i
\rfloor\}$, i.e., if  $  \Delta_1 <  t_i \le \Delta_K $, we set
$\gamma_i = 1$, otherwise $\gamma_i =\lfloor \Delta_1/t_i \rfloor$.

\paragraph{\bf Naming bins}  It is also convenient to name the bins by
groups:
 \[
    \{(i)|\phi_i = 0, 1 \le i \le k\}
 \]
 \[
   \{(i,?)|\phi_i \ne 0, 1 \le i \le k\}
 \]
 \[
   \{(?,j)|\alpha_j \ne 0, 1 \le j \le k\}
 \]
\[
    \{(i,j)|\phi_i \ne 0, \alpha_j \ne 0, \gamma_j t_j \le \Delta_{\phi(i)},
1 \le i,j \le k\}.
 \]

Group $(i)$ consists of bins that hold only blue items of
type-$i$. Group $(i,j)$ consists of bins that contain blue items of
type-$i$ and red items of type-$j$. Blue group $(i,?)$ and
red group $(?,j)$ are indeterminate bins  {\em
currently} containing only blue items of type-$i$ or red
items of type-$j$ respectively. During packing, red items or
blue items will be packed into indeterminate bins if necessary, i.e., indeterminate
bins will be changed into $(i,j)$.

The Super Harmonic algorithm is outlined below:

{\bf Super Harmonic }
   \begin{enumerate}
%%   \item  Initialize $e_i \gets 0$ and $s_i \gets 0$ for $1 \le i \le k$.
   \item  For each item $p$ :  $i \gets$ type of $p$,
        \begin{enumerate}
          \item if $i = k+1$ then use NF algorithm,
           \item else $s_i \gets s_i +1$;
             if $e_i < \lfloor \alpha_i s_i \rfloor$
                  then $e_i \gets e_i +1$; \{ color $p$  red \}
                \begin{enumerate}
%%                  \item $e_i \gets e_i +1$; \{ color p red \}
                  \item If there is a bin in group $(?,i)$ with fewer
                   than $\gamma_i$ type-i items,
                   then place $p$ in it.  \\
                   Else if, for any $j$, there is a bin in group $(j,i)$
                   with fewer than $\gamma_i$  type-i items then
                   place $p$ in it.
                   \item  Else if there is some bin in group $(j,?)$ such
                         that $\Delta_{\phi(j)} \ge \gamma_i t_i $, then
                          pack $p$ in it and
                        change  the bin into $(j,i)$.
                  \item Otherwise, open a bin $(?,i)$, pack $p$ in it.
                \end{enumerate}
          \item else \{color $p$ blue\}:
             \begin{enumerate}
             \item if $\phi_i =0$ then
 %                  \begin{enumerate}
                       if there is a bin  in group $i$ with
                            fewer than $\beta_i$ items then pack $p$ in it,
                            else  open a new group $i$ bin,
                           then pack $p$ in it.
 %                  \end{enumerate}
              \item Else:
                   \begin{enumerate}
                    \item if, for any $j$, there is a bin in
                         group $(i,j)$  or $(i,?)$
                         with fewer than $\beta_i$ type-i items,
                         then pack $p$ in it.
%%                          Else if there is an open bin in group $(i,?)$ with fewer
%%                          than $\beta_i$ type-i items,
%%                           then pack $p$ in it.

                    \item  Else if there is a bin in group $(?,j)$ such
                           that $\Delta_{\phi(i)} \ge \gamma_j t_j$ then
                              pack $p$ in it,
                            and change the group of this bin into $(i,j)$.
                    \item Otherwise, open a new bin $(i,?)$ and
                           pack $p$ in it.
                     \end{enumerate}
             \end{enumerate}
      %%       \end{enumerate}
        \end{enumerate}
   \end{enumerate}

\section{New Weighting Functions for Super Harmonic}
In this section, we develop new weighting functions for Super
Harmonic that are simpler than the weighting system in \cite{S02}.
The weighting functions will be useful in analyzing the proposed online 
algorithm as we shall see in the next section.

\subsection{Intuitions for defining weights}
Weighting functions are widely used in analyzing online bin packing problems.
Roughly speaking, for an item, the value by one of weight functions
is the fraction of a bin occupied by the item in the online algorithm. 
There is a constraint in defining weights for items for an online algorithm.
Let $K+1$ be the number of weighting functions.
Let $W^{i}(p)$ be the weight of an item $p$, where $1 \le i \le K+1$.
For any input $L$, the constraint is
\begin{equation} \label{eqn:weight-constraint}
A(L)\le \max_{1 \le i \le K+1}\big\{ \sum_{p \in L}W^{i}(p) \big\} + O(1),
\end{equation}
where $A(L)$ is the number of bins used by algorithm $A$.

Consider  Super Harmonic algorithm.
For $1 \le i \le k$, let $l_i$ be the number of type-$i$ pieces.
For $1 \le i, s \le k$, 
let $B_{(i)}, B_{(i,s)}, B_{(i,?)}, B_{(?,i)}$ be the number of bins in 
groups $(i)$, $(i,s)$, $(i,?)$ and $(?,i)$. 
Then we have 
%%Then we have
\begin{equation} \label{eqn:blue}
 \sum_{i} \Big\{ B_{(i)} +  \sum_s B_{(i,s)} + B_{(i,?)}\Big\} = \sum_{i} \frac{(1-\alpha_i)l_i}{\beta_i} +O(1)
\end{equation}
and
\begin{equation} \label{eqn:red}
  \sum_{i}\Big\{ B_{(?,i)} + \sum_s B_{(s,i)}\Big\} = \sum_{i} \frac{\alpha_i l_i}{\gamma_i} + O(1).
\end{equation}
So, for each item with size $x \in I_i$, where   $i \le k$,
if we define its weight as below:
\[
  \frac{1-\alpha_i}{\beta_i} + \frac{\alpha_i}{\gamma_i},  
\]
then it is not difficult to see that the constraint (\ref{eqn:weight-constraint}) holds. 
However  the above weighting function is not good enough, i.e., 
it  always leads a competitive ratio at least 1.69103.

The main reason is that for each bin in group $(i,s)$ we account it  twice, 
where $1 \le i, s \le k$.
Next we give some intuitions for improving the above weighting function.

By (\ref{eqn:blue}) and (\ref{eqn:red}), observe that 
\begin{equation} \label{eqn:blue1}
 \sum_{i}  \sum_s B_{(i,s)}  \le \sum_{i} \frac{(1-\alpha_i)l_i}{\beta_i} +O(1)
\end{equation}
and
\begin{equation} \label{eqn:red1}
  \sum_{i} \sum_s B_{(s,i)} \le \sum_{i} \frac{\alpha_i l_i}{\gamma_i} + O(1).
\end{equation}
So, we have 
\begin{eqnarray*}
 \sum_{i} \sum_s B_{(i,s)} = \frac{\sum_{i} \sum_s B_{(i,s)} + \sum_{i} \sum_s B_{(s,i)}}{2} &\le& \sum_{i} \frac{(1-\alpha_i)l_i}{2\beta_i} + \sum_{i} \frac{\alpha_i l_i}{2\gamma_i} +O(1) \\
  &=& l_i \sum_{i} \Big (\frac{1-\alpha_i}{2\beta_i} + \frac{\alpha_i }{2\gamma_i} \Big ) +O(1). 
\end{eqnarray*}
Hence, for an item with size $x \in I_i$,  after packing, 
if there is  a bin in  group $(i,s)$ and also a bin in group $(s,i)$, 
then we can define its weight as below:
\[
  \frac{1-\alpha_i}{2\beta_i} + \frac{\alpha_i}{2\gamma_i}.  
\]
This is the main intuition to lead our weighting functions,
which are given in the next subsection.

\subsection{New weighting functions}
Remember that in Super Harmonic, there is a set $D=\{\Delta_0,
\Delta_1,...,\Delta_K\}$ representing the ``free spaces'' reserved for red
items. 
Recall the two functions $\phi(i)$ and $\varphi(i)$ are related to free spaces
and have the meanings as below:
$\phi(i) =j$ implies that free space $\Delta_j$ is reserved
for red items in a bin consisting of $\beta_i$ blue items of type-$i$,
and $\varphi(i) =j$ indicates that a red item of type-$i$ could be
packed in free space $\Delta_{\ge j}$.

We are now ready to define new weighting functions. Items with size
larger than $\epsilon$ will be first considered. Let $E$ be the number
of indeterminate red group bins $(?, i)$ when the whole packing is
done. 

If $E = 0$, i.e., every red item is placed in a bin with one or more
blue items, then we define the weighting function as:
\begin{equation}
  W^1(x) = \frac{1-\alpha_i}{\beta_i}  \qquad \textrm{ if } x \in I_i.
\end{equation}

Otherwise,  $E >0$ implying that for some $i$, an indeterminate red group bin
$(?, i)$ exists after packing. Let $e$ be the smallest red item in
 indeterminate red group bins. Assume $r$ is the type of item $e$ and $j
= \varphi(r)$. If $2 \le j \le K $ then we define the corresponding
weighting functions as follows:
   \begin{displaymath}
  \begin{array}{ll}
 W^{K+2-j}(x)  =  &   \left\{
                                    \begin{array}{ll}
                 \frac{1-\alpha_i}{\beta_i} + \frac{\alpha_i}{2\gamma_i}     &  \textrm{if } x \in I_i   \textrm{ } \phi(i) < j, \textrm{ and } \varphi(i) < j \\
\vspace{4pt}
            \frac{1-\alpha_i}{\beta_i} + \frac{\alpha_i}{\gamma_i}     &  \textrm{if } x \in I_i   \textrm{ } \phi(i) < j, \textrm{ and } \varphi(i) \ge j \\
\vspace{4pt}
         \frac{1-\alpha_i}{2\beta_i} +    \frac{\alpha_i}{\gamma_i}     &  \textrm{if } x \in I_i   \textrm{ } \phi(i) \ge j, \textrm{ and } \varphi(i) \ge j \\
\vspace{4pt}
             \frac{1-\alpha_i}{2\beta_i} +  \frac{\alpha_i}{2\gamma_i}     &  \textrm{if } x \in I_i   \textrm{ } \phi(i) \ge j, \textrm{ and } \varphi(i) < j
                                    \end{array}
                             \right.
  \end{array}
 \end{displaymath}

If $j =1$, we define
  \begin{displaymath}
  \begin{array}{ll}
  W^{K+1}(x)  =  &   \left\{
                                    \begin{array}{ll}
                 \frac{1-\alpha_i}{\beta_i}      &  \textrm{if } x \in I_i   \textrm{ } \phi(i) =0, \textrm{ and } \varphi(i) =0 \\
\vspace{4pt}
 \frac{1-\alpha_i}{\beta_i} + \frac{\alpha_i}{\gamma_i}     &  \textrm{if } x \in I_i   \textrm{ } \phi(i) =0, \textrm{ and } \varphi(i) >0  \\
\vspace{4pt}
  0     &  \textrm{if } x \in I_i   \textrm{ } \phi(i) > 0, \textrm{ and } \varphi(i) =0  \\
\vspace{4pt}
           \frac{\alpha_i}{\gamma_i}     &  \textrm{if } x \in I_i   \textrm{ } \phi(i) > 0 \textrm{ and } \varphi(i) > 0
                                    \end{array}
                             \right.
  \end{array}
 \end{displaymath}
Note that in the above definitions, if $\gamma_i =0$ then we replace
$\frac{\alpha_i}{\gamma_i}$ with zero. For an item with size $x \in
I_{k+1}$, we always define $W^{j} (x) = \frac{x}{1-\epsilon}$ for
all $j$.

\begin{theorem} \label{th:SHW} For any list $L$, we have
\[
 A(L) \le \max_{1 \le i \le K+1}\big\{ \sum_{p \in L}W^{i}(p) \big\} + O(1),
\]
where $A(L)$ is the number of bins used by Super Harmonic for list
$L$.
\end{theorem}
\begin{proof}
Fix a list $L$. Let $D$ be the sum of the sizes of the items of 
type-$(k+1)$. By NEXT FIT, we know that the number of bins used for 
type-$(k+1)$ is at most $D/(1-\epsilon) + 1$. 

Again, we use $E$ to denote
the number of indeterminate red group bins when all the packing is
done.
If  $E >0$, let $e $ be the smallest red item in  indeterminate red
group bins. Assume $r$ is the type of item $e$ and $j = \varphi(r)$.
 For $1 \le i \le k$, let $l_i$ be the number of type-$i$ pieces.
Let $B_{(i)}, B_{(i,s)}, B_{(i,?)}, B_{(?,i)}$ be the number of bins in 
groups $(i)$, $(i,s)$, $(i,?)$ and $(?,i)$.

To prove this theorem, we consider three cases.

Case 1:
If $E =0$, i.e., $\sum_i B_{(?,i)} =0$, 
every red item is packed in a bin with one or more blue items.
Therefore we just need to count bins containing blue items:
\begin{eqnarray*}
A(L) &\le& \frac{D}{1-\epsilon} +  \sum_{i} \Big\{ B_{(i)} +  \sum_s B_{(i,s)} + B_{(i,?)}\Big\} +O(1) \\
&\le&  \sum_{x \in I_{k+1}} W^{1}(x) + \sum_{i} \frac{(1-\alpha_i)l_i}{\beta_i} +O(1)  \quad \textrm{by  (\ref{eqn:blue})}\\
&=& \sum_{x \in I_{k+1}, x \in L} W^{1}(x) + \sum_{ x \notin
I_{k+1},x \in L} W^{1}(x)  +O(1).
\end{eqnarray*}

Case 2: $E >0$,  $e $ is  the smallest red item in indeterminate
red group bins and its type is $r$ and $\varphi(r) = j \ge 2$. Since
every red item of type-$i$ is placed in a final group bin $(s,i)$,
where $\varphi(i) < j$, we have
\begin{equation} \label{eqn:red-zero}
 \sum_{\varphi(i) < j} B_{(?,i)} =0.
\end{equation}
On the other hand, we have
\begin{equation} \label{eqn:blue-zero}
  \sum_{\phi(i) \ge j} B_{(i,?)} =0;
\end{equation}
otherwise, $e$ would have been placed into a bin $(i, ?)$,
where $\phi(i) \ge j$.
According to the Super Harmonic algorithm, for any type bin $B_{(i)}$,
we have 
\begin{equation} \label{eqn:phi-zero}
 \phi(i) = 0.
\end{equation}

Define 
\[
 X = \sum_{\substack{ \phi(i) \ge j \\ \varphi(s) < j}} B_{(i,s)},
\]
which is  the total number of  all the bins in groups $(i,s)$ such that
$\phi(i) \ge j$  and $\varphi(s) < j $.
Then we have 
\begin{eqnarray}
  A(L) &\le& \frac{D}{1-\epsilon} +  \sum_i \Big(B_{(i)} + B_{(i,?)} + B_{(?,i)}\Big) + \sum_{i}\sum_s  B_{(i,s)} + O(1) \nonumber \\
&= & \frac{D}{1-\epsilon} +  \sum_i \Big(B_{(i)} + B_{(i,?)} + B_{(?,i)}\Big) + X + \sum_{\phi(i) < j } \sum_s  B_{(i,s)}+ \sum_{\varphi(i) \ge j } \sum_s  B_{(s,i)} + O(1) \nonumber \\
&= & 
 \frac{D}{1-\epsilon} +  \sum_{\phi(i) <j} \Big(B_{(i)} + B_{(i,?)}+ \sum_s  B_{(i,s)}\Big)  + \sum_{\varphi(i) \ge j }\Big(B_{(?,i)} + \sum_s  B_{(s,i)}\Big )+X + O(1). \label{eqn:B-X}
\end{eqnarray}
The last inequality follows directly from (\ref{eqn:red-zero}), (\ref{eqn:blue-zero}) and (\ref{eqn:phi-zero}).

Then by  the  definition of variable $X$, we have
\[
  X \le \sum_{j \le  \phi(i) \le K} \sum_s B_{(i,s)}
  \textrm{ and }
 X \le \sum_{1 \le  \varphi(i) \le j-1} \sum_s B_{(s,i)}.
\]
Therefore, 
\begin{equation} \label{eqn:averageX}
 X \le \Big\{\sum_{j \le  \phi(i) \le K} \sum_s B_{(i,s)} + \sum_{1 \le  \varphi(i) \le j-1} \sum_s B_{(s,i)}\Big\}/2.
\end{equation}
So, by (\ref{eqn:B-X}) and (\ref{eqn:averageX}),
 we have
%%\begin{displaymath}

\begin{eqnarray*}
  A(L) %%&\le& %% \frac{D}{1-\epsilon} +  \sum_{\phi(i) <j}
%%  (B_{(i)}+ B_{(i,?)} + B_{(i,s)})  +\sum_{\varphi(i) \ge j} (R_{(?,i)} + R_{(s,i)})
%%   + \bar{ B} +O(1)\\
 &\le& \frac{D}{1-\epsilon}  +   \sum_{\phi(i) <j} \Big(B_{(i)} + B_{(i,?)}+ \sum_s  B_{(i,s)}\Big)  + \sum_{\varphi(i) \ge j }\Big(B_{(?,i)} + \sum_s  B_{(s,i)}\Big )\\
&& +  \sum_{ \phi(i) \ge j } \sum_s \frac{ B_{(i,s)} }{2}
+  \sum_{  \varphi(i) < j } \sum_s \frac{ B_{(s,i)}}{2}
+ O(1)\\
 &\le & \frac{D}{1-\epsilon}
  + \sum_{\phi(i) < j} \frac{(1-\alpha_i)l_i}{\beta_i}
  + \sum_{\varphi(i) \ge j} \frac{\alpha_i l_i}{\gamma_i}
  +  \sum_{ \phi(i)  \ge  j} \frac{(1-\alpha_i)l_i}{2\beta_i}
  +  \sum_{\varphi(i) < j} \frac{\alpha_i l_i}{2\gamma_i} + O(1)\\
&\le&  \frac{D}{1-\epsilon}
  + \sum_{\substack{ \phi(i) < j \\ \varphi(i) < j}} \Big(\frac{(1-\alpha_i)l_i}{\beta_i} + \frac{\alpha_i l_i }{2\gamma_i} \Big)
 + \sum_{\substack{ \phi(i) < j \\ \varphi(i) \ge j}}\Big( \frac{(1-\alpha_i)l_i}{\beta_i} + \frac{\alpha_i l_i }{\gamma_i} \Big) \\
&& + \sum_{\substack{ \phi(i) \ge j \\ \varphi(i) \ge j}} \Big(\frac{(1-\alpha_i)l_i}{2\beta_i} + \frac{\alpha_i l_i }{\gamma_i} \Big)
+ \sum_{\substack{ \phi(i) \ge j \\ \varphi(i) < j}} \Big(\frac{(1-\alpha_i)l_i}{2\beta_i} + \frac{\alpha_i l_i }{2\gamma_i} \Big) + O(1) \\
    &= & \sum_{x \in I_{k+1}, x \in L} W^{K+2-j}(x) + \sum_{ x \notin I_{k+1}, x \in L} W^{K+2-j}(x)  +O(1)
\end{eqnarray*}
The second inequality follows directly 
from (\ref{eqn:blue}) and (\ref{eqn:red}).
%%\end{displaymath}

Case 3. $E >0$ and $j =1$. The arguments are analogous with Case 2.
According to the Super Harmonic algorithm, 
for  any type of bin $(i,s)$, we have $\varphi(s) \ge 1$, 
where $1 \le i, s \le k$ and $k$ is a parameter defined in Super Harmonic. 
So, there is no such bin $(i,s)$ with $\varphi(s) < 1$.
Then we have 
\begin{eqnarray*}
A(L) &\le& \frac{D}{1-\epsilon} +  \sum_{\phi(i) < 1} \Big(B_{(i)}+
B_{(i,?)} + \sum_s B_{(i,s)}\Big)  +\sum_{\varphi(i) \ge 1} \Big(B_{(?,i)} +\sum_s B_{(s,i)}\Big)+O(1)\\
 &\le & \frac{D}{1-\epsilon}
  + \sum_{\phi(i) =0} \frac{(1-\alpha_i)l_i}{\beta_i}
  + \sum_{\varphi(i) \ge 1} \frac{\alpha_i l_i}{\gamma_i}
  + O(1)\\
&\le&  \frac{D}{1-\epsilon}
 + \sum_{\substack{ \phi(i) =0 \\ \varphi(i)  =0}} \frac{(1-\alpha_i)l_i}{\beta_i}
  + \sum_{\substack{ \phi(i) =0 \\ \varphi(i) > 0}}  \Big(\frac{(1-\alpha_i)l_i}{\beta_i} + \frac{\alpha_i l_i }{\gamma_i} \Big)
 + \sum_{\substack{ \phi(i) > 0 \\ \varphi(i) > 0}} \frac{\alpha_i l_i }{\gamma_i} +O(1)\\
    &= & \sum_{x \in I_{k+1}, x \in L} W^{K+1}(x) + \sum_{ x \notin I_{k+1}, x \in L} W^{K+1}(x)  +O(1)
\end{eqnarray*}

Therefore, we have $A(L) \le \max_{1 \le i \le K+1}\big\{ \sum_{p
\in L}W^{i}(p) \big\} + O(1)$.
\end{proof}

\section{Algorithm $H \otimes SH+$ and Its Analysis}
In the section, we first review a class of online algorithms for
two dimensional online bin packing, called $H \otimes B$
\cite{SS03}. Next we introduce a new instance of algorithm $H \otimes
SH+$, where $H$ is Harmonic and $SH+$ (Strange Harmonic+) is an
instance of Super Harmonic. Then we propose some new techniques  on
how to bound the total weight in a single bin, which is crucial to
obtaining a better asymptotic competitive ratio for the $H \otimes B$
algorithm. Finally, we apply new weighting functions for $SH+$ to
analyze the two-dimensional online bin packing algorithm $H \otimes
SH+$ and show its competitive ratio  at most 2.5545, which
implies that the new weighting functions work very well with the
generalized approach of bounding the total weight in a single bin.
Note that as mentioned in \cite{SS03} if we apply the weighting
functions of $SH+$ derived from \cite{S02} directly to analyze
algorithm $H \otimes SH+$ then the upper bound cannot be improved.

\subsection{Algorithms $H \times B$ and $H \otimes B$}
Now we review two-dimensional online bin packing algorithms $H
\times B$ and $H \otimes B$ \cite{SS03}, where $H$ is Harmonic and
$B$ is Super Harmonic.

Given an item $p =(w,h)$, $H \times B$ operates as follows:
\begin{enumerate}
\item {\bf Packing items into slices}:
       If $w \ge \epsilon$ then pack $p$ into a slice of height 1 and width $t_i$
       by $H$ (Harmonic algorithm), where $t_{i+1} < w \le t_i$;
       else pack it into a slice of height 1 and width $ \epsilon(1-\delta)^i$
       by $H$ (Harmonic algorithm), where $\epsilon (1-\delta)^{i+1} < w \le \epsilon (1-\delta)^i$ and $\delta >0$ is  arbitrarily small.
\item {\bf Packing slices into bins}:
      When a new slice is required in the above step,
      we allocate it from a  bin using algorithm $B$.
\end{enumerate}

$H \otimes B$ is a randomized algorithm, which operates as follows:
before processing begins, we flip a fair coin. If the result is
heads, then we run $ H \times B$; otherwise we run $B \times H$, i.e.,
the roles of height and width are interchanged. Note that it is
possible to de-randomize $H \otimes B$ without increasing its
performance ratio. For details, we refer to \cite{SS03}.

\begin{theorem}\label{th:ratio}
If an online 1D bin packing algorithm $B$ has weighting functions
$W_B^{i}(x)$ such that $B(L) \le \max_i\{\sum_{x \in L}
W_B^{i}(x)\} + O(1).$ Then  the cost by algorithm $H\otimes B$ for input $L$
is at most
\[
 \frac{1}{2(1-\delta)} \Big( \max_i \big\{ \sum_{p \in L} W_{H \times B}^i(p)\big\}+
\max_i \big\{\sum_{p \in L} W_{B \times H}^i(p)\big\}
\Big) +O(1),
\]
and the asymptotic competitive ratio of
algorithm $H\otimes B$ is at most
\[
\frac{1}{2(1-\delta)}\max_{\forall X}\Bigg(\max_{ i}\bigg\{  \sum_{(x,y) \in X} W_H(x) W_B^{i}(y) \bigg\} +
   \max_{i}\bigg\{ \sum_{(x,y) \in X} W_H(y) W_B^{i}(x),  \bigg\}\Bigg)
\]
where $\delta$ is a parameter defined in $H \otimes B$ algorithm
and $X$ is a set of items which fit in a single bin.
\end{theorem}

\subsection{ An instance of Super Harmonic $SH+$} \label{subsec:SH+}
 As mentioned in \cite{S02}, 
 it is a hard problem to find appropriate  parameters
 in designing an instance of Super Harmonic, especially setting $t_i$.
 The parameters in $SH+$ are found through a trial-and-error way
 and are defined as follows:
\begin{displaymath}
 \begin{array}{|c|c|c|c|c|c|c|c|}
  \hline
    i & t_i & \alpha_i & \beta_i &  \delta_i & \phi(i) & \varphi(i) & \gamma_i \\
  \hline
   1  &  1   & 0 & 1  &  0      &   0  & 0 &0     \\
   2  &  0.706 & 0  & 1 & 0.294   &  1   & 0  &0  \\
 %% \hline
   3  &  0.657 & 0 & 1  &  0.343    &   2  &   0 &0 \\
   4  &  0.647 & 0 & 1  &  0.353    &   3 &   0  &0  \\
 %%  \hline
   5  &  0.625 & 0& 1  &  0.375  &  4   &  0 &0 \\
   6  & 0.6   & 0 & 1   &  0.4   &  5    &  0 & 0 \\
%%  \hline
  7   &   0.58  & 0 & 1  &  0.42   &  6   &  0&  0  \\
%% \hline
  8   &  0.5  & 0 & 2   &  0     & 0    &  0 &  0 \\
%%  \hline
 9    & 0.42  & 0.162 & 2   &  0.16   & 0    &  6 &  1   \\
 10   & 0.4   & 0.192 & 2  & 0.2   &   0   &  5 &    1  \\
 11    & 0.375 & 0.2346 & 2 & 0.25 &   0  &  4&    1  \\
%%  \hline
 12    &   0.353 & 0.3004 & 2  &  0.294   & 1   &  3 &  1   \\
 13    &   0.343 & 0.3077 & 2  &  0.314   & 1    &  2 &  1   \\
%% \hline
% \end{array}
% \end{displaymath}
% \begin{displaymath}
%  \begin{array}{|c|c|c|c|c|c|c|c|}
% \hline
%   i & t_i & \alpha_i & \beta_i &  \delta_i & \phi(i) & \varphi(i) & \gamma_i \\
%   \hline
 14    & 1/3  & 0 & 3  &  0     & 0    &  0   &  0 \\
%%\hline
 15    &  0.294 & 0.0816 & 3  &  0.118 & 0   &  1 &  1 \\
%%\hline
 16   &  1/4   & 0.186 & 4  &  0     & 0   &  1 &  1 \\
%%%\hline
 17   &  1/5   & 0.092 & 5  &  0     & 0   &  1 &  1  \\
%%\hline
 18   &   1/6  &0.1456 & 6  &  0     & 0    &  1 &  1  \\
 19   &  0.147 & 0.2162 & 6   & 0.118       &0    &  1  & 2 \\
 20   &  1/7   & 0.1525 & 7   & 0      &0    &  1  & 2 \\
 21-49 &  1/(i-13) & ff(i)  & i-13  &  0     & 0   &  1 &  \lfloor\Delta_1 /t_i \rfloor  \\
50   & 1/37   & 0   & 37  &  0     & 0   &  0 &  0  \\
 51   &    1/38 & 0 & * &  *     & *     &  * &  *  \\
\hline
 \end{array}
\hspace{3mm}
\begin{array}{|c|c|c|}
 \hline
  j = \phi(i) & \Delta_j  &  \textrm{ Red accepted} \\
 \hline
%%  1 &  0.20    &  11..16   \\
  1 &  0.294    &  15..50   \\
  2 &  0.343    &  13, 15..50  \\
  3 &  0.353    &  12,13, 15..50   \\
  4 &  0.375    &  11..13, 15..50  \\
  5 &  0.4      &  10..13, 15..50   \\
  6 &  0.42     &   9..13,  15..50    \\
 \hline
\end{array}
 \end{displaymath}
 where $ff(i) = 1.35(50-i)/37(i-12)$.

Then  we have seven  weighting functions for $SH+$, i.e.,
$W_{B}^{i}$  as defined in the last section, where  $1 \le i
\le 7$.

\subsection{Previous framework for calculating upper bounds}
In this subsection, we first introduce the previous framework
for computing the upper bound of the competitive ratio of 
 $H \otimes SH+$, then mention that the previous framework does not 
work well with the instance in the last subsection, i.e., 
the previous framework does not  lead a better upper bound. 

Let $p=(x,y)$ be an item.
We define the following functions.
\[
  W_{H\times B}^{i} (p)= W_H(x)  W_{B}^{i}(y), \textrm{   }
  W_{B \times H}^{i} (p)= W_H(y) W_{B}^{i}(x),
\]
and
\[
W^{i,j} (x,y)= \frac{W_H(x)W_{B}^{i}(y) +  W_{B}^{j}(x) W_H(y)}{2}.
\]

Then we  can   obtain an upper bound on the competitive ratio  $R$ of algorithm $H \otimes SH+$ as follows by Theorems \ref{th:SHW} and  \ref{th:ratio},
 where $X$ is a set of  items which fit in a single bin.
% Then using the similar approach as the one in \cite{SS03},
% by Lemma \ref{lemma:newSH+}, we have

\begin{eqnarray}
%%\begin{array}{ccc}
  R &\le& \frac{1}{2(1-\delta)}\max_{\forall X}\Bigg(\max_{1 \le i \le 7}\bigg\{  \sum_{p \in X} W_{H\times B}^{i}(p)\bigg\} +
   \max_{1 \le i \le 7}\bigg\{ \sum_{p \in X} W_{B\times H}^{i}(p) \bigg\}\Bigg) \nonumber \\
 &\le&  \frac{1}{(1-\delta)}\max_{1 \le i, j \le 7,\forall X }\bigg\{ \sum_{p \in X}(W_{H \times B}^{i}(p) + W_{B \times H}^{j}(p))/2\bigg\} \nonumber \\
 &=&  \frac{1}{(1-\delta)}\max_{1 \le i, j \le 7,\forall X }\bigg\{  \sum_{p \in X}W^{i,j} (x,y)\bigg\} \label{eqn:final-ratio}
%%\end{array}
\end{eqnarray}

The value of $R$ can be estimated by the following approach.
\begin{defi}
Let $f$ be a function mapping from $(0,1]$ to $\mathbb{R}^+$.
$\mathcal{P}(f)$ is the mathematical program: maximize $\sum_{x \in
X}f(x)$ subject to $\sum_{x \in X} \le 1$, over all finite sets of
real numbers $X$. We also use  $\mathcal{P}(f)$ to denote the value
of this mathematical program.
\end{defi}

\begin{lemma}\label{lemma:SS03} \cite{SS03}
Let $f$ and $g$ be functions mapping from $(0,1]$ to $\mathbb{R}^+$.
Let $F=\mathcal{P}(f)$ and $G =\mathcal{P}(g)$. Then the maximum of
$\sum_{p \in X} f(h(p))g(w(p))$ over all finite  multisets of items
$X$ which fit in a single bin is at most $FG$, where $p$ is a
rectangle and $h(p)$ and $w(p)$ are its height and width,
respectively.
\end{lemma}

In \cite{SS03}, $f$ and $g$ are defined as below:

\[
  f^{i,j} (y)=   \frac{W_H(y) +  W_{B}^{i}(y)}{2},
\]
 and
\[
  g^{i,j}(x) = \sup_{0 < y \le 1} \frac{ W^{i,j}(x,y)}{f^{i,j}(y)}.
\]
By the above definitions, we have 
\[
 W^{i,j}(x,y) \le f^{i,j} (y)g^{i,j}(x),
\]
for all $0 \le x \le 1$ and $0 \le y \le 1$.

{\bf Remarks}: As mentioned in \cite{SS03},
the old weighting functions \cite{S02} do not work well with
the calculating framework used in \cite{S02, SS03},
i.e., the previous framework with the old weighting function
does not guarantee  an upper bound better than 2.66013.

\subsection{A new framework for calculating upper bound}
In this subsection, we first generalize the previous analysis
framework by introducing a new lemma and developing new functions for $f$ and $g$ 
in order to bound the total weight in a single bin. 
Then we apply our new weighting functions for   Super Harmonic to
algorithm $H \otimes SH+$ and obtain a new upper bound for two-dimensional online
bin packing.

\begin{lemma}\label{lemma:transpose}
$\max_{\forall X }\bigg\{  \sum_{p \in X}W^{i,j} (x,y)\bigg\}
 = \max_{\forall X }\bigg\{  \sum_{p \in X}W^{j,i} (x,y)\bigg\}
$, where $1 \le i, j \le 7$
\end{lemma}
\begin{proof}
 By definition, observe that for any  $1 \le i, j \le 7$,
\begin{equation} \label{eqn:xyyx}
  W^{i,j}(x,y) = W^{j,i} (y,x).
\end{equation}
Let $X =\{p_1, p_2, \dots, p_m\}$ be a set of rectangles which fit
into  a single bin, where $p_i = (x_i,y_i)$ is the $i$-th rectangle in
$X$. If we exchange roles of $x$ and $y$ of $p_i$ to get new
rectangles $p_i' =(y_i,x_i)$ for all $i$, then it is not difficult
to see that the new set $ X' =\{p_1', p_2', \dots, p_m'\}$ is also a
feasible pattern, i.e., all items can fit in a single bin. On the
other hand, by equation (\ref{eqn:xyyx}), we have
\[
  \sum_{p \in X}W^{i,j} (p)
 =  \sum_{p' \in X'}W^{j,i} (p'),
\]
 where $1 \le i, j \le 7$.
There is a one-to-one mapping between $X$ and $X'$ in all the feasible patterns.
Therefore, we have this lemma.
\end{proof}

\noindent{\bf New functions $f$ and $g$}:
We define new functions $f$ and $g$ such that 
(i) Lemma  \ref{lemma:SS03} can be applied to bound the weight in a single bin,
and (ii) the resultant bound is not too loose.
The new functions $f$ and $g$ are defined as follows:
\[
  f^{i,j} (y)=  \lambda_{i,j} W_H(y) +  (1-\lambda_{i,j})W_{B}^{i}(y),
\]
where $0 \le \lambda_{i,j} \le 1$ and
\[
  g^{i,j}(x) = \sup_{0 < y \le 1} \frac{ W^{i,j}(x,y)}{f^{i,j}(y)}.
\]
Note that in \cite{SS03}, $\lambda_{i,j}$ are 1/2 for all $i,j$.
It is not difficult to see that the following inequality still holds
although we have generalized the definition of the $f$ function,
\[
 W^{i,j}(x,y) \le f^{i,j} (y)g^{i,j}(x)
\]
for all $0 \le x \le 1$ and $0 \le y \le 1$.

\smallskip

\noindent{\bf  New approach of calculating $\mathcal{P}(f)$}:
In order to  use   Lemma \ref{lemma:SS03}  to obtain the upper bound on
the competitive ratio $R$ of algorithm $H \otimes B$,
we need to calculate $\mathcal{P}(f^{i,j})$ and
$\mathcal{P}(g^{i,j})$.
 Let $f$ be one of   $f^{i,j}$ or  $g^{i,j}$
for $1 \le i, j \le 7$.
In \cite{S02}, Seiden wrote a programming to  enumerate all the 
feasible patterns to get the bounds for  $\mathcal{P}(f)$. 
Here, we give a simple approach by calling LP solver directly 
to estimate  $\mathcal{P}(f)$, which can be modeled  as the
following mixed integer program (MIP):
     \begin{eqnarray*}
 \textrm{max. \qquad } f &=& \sum_{i=1}^{50}x_i w_i + (1- \sum_{i=1}^{50} x_i (t_{i+1} + \epsilon) ) \times \frac{1}{1-t_{51}} \qquad (1)\\
       \textrm{s.t. \qquad }
                      && \sum_{i=1}^{50} x_i (t_{i+1} + \epsilon ) \le 1, \\
                      && x_i \ge 0, \textrm{ integer}.
  \end{eqnarray*}
where $x_i$ is the number of type-$i$ items in a feasible pattern, $w_i$ is
the weight for an  item of type-$i$, which is decided by function $f$,
i.e., $w_i =f(p)$ if $p \in (t_{i+1},t_i]$. Since $\epsilon >0$ can
be arbitrarily small, we cannot  find an exact value for $\epsilon$.
Therefore, we set $\epsilon =0$ and re-model the above MIP as
follows.
%%     \begin{eqnarray*}
\begin{displaymath}
\begin{array}{ccc}
 \textrm{max. \qquad } f &=& \sum_{i=1}^{50}x_i w_i + (1- \sum_{i=1}^{50} x_i t_{i+1} ) \times \frac{1}{1-t_{51}} \qquad (2) \\
\end{array}
\end{displaymath}
\begin{displaymath}
 \begin{array}{lll}
       \textrm{s.t. \qquad }
                       \sum_{i=1}^{50} x_i t_{i+1}  \le 1, &&    \\
                       x_i \le 1, \textrm{ for }1 \le i \le 7, && 
                       x_i \le 2, \textrm{ for } 8 \le i \le 13, \\
                       x_i \le 3, \textrm{ for } 14 \le i \le 15,  &&
                        x_i \le i-12, \textrm{ for } 16 \le i \le 17, \\
                        x_{18} + x_{19} \le 6,  &&
                        x_i \le i-13, \textrm{ for } 20 \le i \le 50, \\
                       2x_7 + x_{15} \le 3.9,   &&
                        3x_7 + 2 x_{13} + x_{17} \le 5.9, \\
                      4x_{13} + 3 x_{15} + x_{24} \le 11.9,  &&  
                        5x_7 + 3.53 x_{11} + 1.47x_{18} \le 9, \\
                       12x_7 + 8 x_{13} + 3x_{20} + x_{36} \le 23, && 
                        9x_7 + 6 x_{13} + 2x_{21} + x_{30} \le 17, \\
                      x_i \ge 0, \textrm{ integer}. &&
   \end{array}
\end{displaymath}
 %% \end{eqnarray*}
Note that the new constraints do not eliminate any feasible solutions
of MIP (1). For example, consider the constraint $5x_7 + 3.53
x_{11} + 1.47x_{18} \le 9$. Since an item of type-7 has size larger
than $0.5$, an item of type-11 has size larger than $0.353$ and an
item of type-18 has size larger than $0.147$, we have $0.5x_7+ 0.353
x_{11} + 0.147x_{18} < 1$. So, we have  $5x_7+ 3.53 x_{11} +
1.47x_{18} < 10$. It is not difficult to see that the following
inequality $ 5x_7 + 3.53 x_{11} + 1.47x_{18} \le 9$ is equivalent
to $5x_7+ 3.53 x_{11} + 1.47x_{18} < 10$ when $x_7$, $x_{11}$ and
$x_{18}$ are non-negative integers. For other constraints in MIP
(2), the arguments are analogous.

To solve MIP (2), we use a tool for solving linear and integer
programs called {\em GLPK} \cite{GLPK}. We write a program to
calculate $W^{i,j}(x,y)$, $g^{i,j}(x)$  and $f^{i,j}(y)$ for each ($i,j$),
and then call API of GLPK  to calculate
$\mathcal{P}(f^{i,j})$ and  $\mathcal{P}(g^{i,j})$. The values of
$\mathcal{P}(f^{i,j})$ and  $\mathcal{P}(g^{i,j})$ are shown in the tables
in Appendix.

Note that when we use Lemma \ref{lemma:SS03} for the upper bound on the
weight $\max_{\forall X}\{ \sum_{p \in X}W^{i,j} (x,y)\}$, for all
pairs $(i,j)$, the calculations are independent. For different pairs
$(i,j)$, $\lambda_{i,j}$ may be different. So, in
order to get an upper bound near the true value of  $\max_{\forall
X}\{ \sum_{p \in X}W^{i,j} (x,y)\}$, we have to select an
appropriate $\lambda_{i,j}$. This can be done by a trial-and-error
approach.

\begin{theorem}\label{th:new}
For all $\delta >0$, the asymptotic competitive ratio of $H \otimes
B$ is at most  2.5545.
\end{theorem}
%%\section{The Proof of Theorem \ref{th:new}}
\begin{proof}
According to the tables in Appendix, by Lemma
\ref{lemma:transpose} and  Lemma \ref{lemma:SS03}, we have
\[
\max_{\forall X }\bigg\{  \sum_{p \in X}W^{1,2} (p)\bigg\}
 = \max_{\forall X }\bigg\{  \sum_{p \in X}W^{2,1} (p)\bigg\}
 \le \mathcal{P}(f^{1,2}) \mathcal{P}(g^{1,2}) \le 2.5539.
\]
\[
\max_{\forall X }\bigg\{  \sum_{p \in X}W^{1,6} (p)\bigg\}
 = \max_{\forall X }\bigg\{  \sum_{p \in X}W^{6,1} (p)\bigg\}
 \le \mathcal{P}(f^{6,1}) \mathcal{P}(g^{6,1}) \le 2.5545.
\]
\[
\max_{\forall X }\bigg\{  \sum_{p \in X}W^{2,5} (p)\bigg\}
 = \max_{\forall X }\bigg\{  \sum_{p \in X}W^{5,2} (p)\bigg\}
 \le \mathcal{P}(f^{5,2}) \mathcal{P}(g^{5,2}) \le 2.5340.
\]
\[
\max_{\forall X }\bigg\{  \sum_{p \in X}W^{2,6} (p)\bigg\}
 = \max_{\forall X }\bigg\{  \sum_{p \in X}W^{6,2} (p)\bigg\}
 \le \mathcal{P}(f^{6,2}) \mathcal{P}(g^{6,2}) \le 2.5364.
\]
For all the other $(i,j)$, by Lemma \ref{lemma:SS03}, we have
\[
\max_{\forall X }\bigg\{  \sum_{p \in X}W^{i,j} (p)\bigg\}
 \le \mathcal{P}(f^{i,j}) \mathcal{P}(g^{i,j})
 \le \mathcal{P}(f^{1,1}) \mathcal{P}(g^{1,1}) \le 2.5545.
\]
 \end{proof}

{\bf Remarks}: 
If we  use the weighting functions from \cite{S02} and the previous analysis
framework, we find that the competitive ratio is at least 3.04. 
(run our programming 2DHSH.c like  ``./2DHSH+.exe old $>$ yourfile'')
Even if we use the new weighting function,  by the previous analysis framework,
the competitive ratio
is still at least 3.04, by running our programming 2DHSH.c like  ``./2DHSH+.exe new1 $>$ yourfile''.
We also find that if we use the old weighting function from   \cite{S02} with
the new analysis framework,  the competitive ratio is at least 2.79.
(run our programming 2DHSH.c like  ``./2DHSH+.exe old2 $>$ yourfile'')
 The reason is that:
 Lemma \ref{lemma:SS03} does not work very well with the old weight function,
 i.e., the resulting value $F\cdot G$ is
away from the maximum weight of items in a single bin.

\section{Concluding Remarks}\label{se:con}
When we use the tool for solving the mixed integer programs,
there are two files which are necessary: one is the model file for
the linear or integer program itself (refer to Appendix), and the other
is the data file where the data is stored. We write a program to
generate the data and then call the tool GLPK. (Actually we call API
(Application Program Interface) of GLPK. To download
the source file, go to:
http://sites.google.com/site/xinhan2009/Home/files/2DHSH.c).

Our  framework can be applied to 3D online bin packing to result in
 an algorithm $H\times H \otimes SH+$ with its competitive ratio
$2.5545 \times 1.69103 (\approx 4.3198)$.

\noindent {\bf Acknowledgments}  The authors wish to thank the
  referees for their useful comments on the earlier draft of
 the paper. Their suggestions have helped improve the presentation
 of the paper.

% \normalsize
 \newpage
\appendix
% \begin{center}
% {\Large \bf  Appendix}
% \end{center}

\section{Values of $f^{i,j}$ and $g^{i,j}$}
\begin{displaymath}
\begin{array}{|c|c|c|c|c|c|c|c|}
\hline
  (i,j) =  &(1,1 )   &(1,2 )   &(1,3 )   &(1,4 )   &(1,5 )   &(1,6 )   &(1,7 )  \\
\hline
\lambda_{i,j} & 0.500000   & 0.500000   & 0.540000   & 0.550000   & 0.565000   & 0.565000   & 0.600000  \\
\hline
\mathcal{P}(f^{i,j}) & 1.598272   & 1.598272   & 1.605095   & 1.606845   & 1.609490   & 1.609490   & 1.615665  \\
\hline
\mathcal{P}(g^{i,j}) & 1.598272   & 1.597872   & 1.574422   & 1.581742   & 1.585430   & 1.587508   & 1.575580  \\
\hline
\mathcal{P}(f^{i,j}) \mathcal{P}(g^{i,j}) & 2.554474   & 2.553834   & 2.527096   & 2.541614   & 2.551734   & 2.555079   & 2.545610  \\
\hline
\end{array}
\end{displaymath}
\begin{displaymath}
\begin{array}{|c|c|c|c|c|c|c|c|}
\hline
  (i,j) =  &(2,1 )   &(2,2 )   &(2,3 )   &(2,4 )   &(2,5 )   &(2,6 )   &(2,7 )  \\
\hline
\lambda_{i,j} & 0.500000   & 0.500000   & 0.530000   & 0.550000   & 0.565000   & 0.565000   & 0.600000  \\
\hline
\mathcal{P}(f^{i,j}) & 1.597328   & 1.597328   & 1.597148   & 1.597028   & 1.596938   & 1.596938   & 1.596729  \\
\hline
\mathcal{P}(g^{i,j}) & 1.609235   & 1.598326   & 1.586301   & 1.595016   & 1.602278   & 1.604268   & 1.589545  \\
\hline
\mathcal{P}(f^{i,j}) \mathcal{P}(g^{i,j}) & 2.570476   & 2.553051   & 2.533557   & 2.547285   & 2.558739   & 2.561917   & 2.538073  \\
\hline
\end{array}
\end{displaymath}
\begin{displaymath}
\begin{array}{|c|c|c|c|c|c|c|c|}
\hline
  (i,j) =  &(3,1 )   &(3,2 )   &(3,3 )   &(3,4 )   &(3,5 )   &(3,6 )   &(3,7 )  \\
\hline
\lambda_{i,j} & 0.500000   & 0.500000   & 0.530000   & 0.550000   & 0.565000   & 0.565000   & 0.600000  \\
\hline
\mathcal{P}(f^{i,j}) & 1.573676   & 1.573676   & 1.572837   & 1.572777   & 1.572732   & 1.572732   & 1.572627  \\
\hline
\mathcal{P}(g^{i,j}) & 1.609235   & 1.598326   & 1.586301   & 1.595016   & 1.602278   & 1.604268   & 1.589545  \\
\hline
\mathcal{P}(f^{i,j}) \mathcal{P}(g^{i,j}) & 2.532414   & 2.515247   & 2.494992   & 2.508604   & 2.519954   & 2.523084   & 2.499762  \\
\hline
\end{array}
\end{displaymath}
\begin{displaymath}
\begin{array}{|c|c|c|c|c|c|c|c|}
\hline
  (i,j) =  &(4,1 )   &(4,2 )   &(4,3 )   &(4,4 )   &(4,5 )   &(4,6 )   &(4,7 )  \\
\hline
\lambda_{i,j} & 0.500000   & 0.500000   & 0.535000   & 0.550000   & 0.565000   & 0.565000   & 0.600000  \\
\hline
\mathcal{P}(f^{i,j}) & 1.581245   & 1.581245   & 1.577140   & 1.575380   & 1.573621   & 1.573621   & 1.569515  \\
\hline
\mathcal{P}(g^{i,j}) & 1.609235   & 1.598326   & 1.586855   & 1.595016   & 1.602278   & 1.604268   & 1.589545  \\
\hline
\mathcal{P}(f^{i,j}) \mathcal{P}(g^{i,j}) & 2.544594   & 2.527344   & 2.502692   & 2.512755   & 2.521378   & 2.524510   & 2.494814  \\
\hline
\end{array}
\end{displaymath}
\begin{displaymath}
\begin{array}{|c|c|c|c|c|c|c|c|}
\hline
  (i,j) =  &(5,1 )   &(5,2 )   &(5,3 )   &(5,4 )   &(5,5 )   &(5,6 )   &(5,7 )  \\
\hline
\lambda_{i,j} & 0.500000   & 0.500000   & 0.535000   & 0.550000   & 0.565000   & 0.565000   & 0.600000  \\
\hline
\mathcal{P}(f^{i,j}) & 1.585370   & 1.585370   & 1.580113   & 1.577860   & 1.575607   & 1.575607   & 1.570350  \\
\hline
\mathcal{P}(g^{i,j}) & 1.609542   & 1.598326   & 1.587240   & 1.595374   & 1.602747   & 1.604737   & 1.589740  \\
\hline
\mathcal{P}(f^{i,j}) \mathcal{P}(g^{i,j}) & 2.551720   & 2.533939   & 2.508019   & 2.517277   & 2.525300   & 2.528436   & 2.496449  \\
\hline
\end{array}
\end{displaymath}
\begin{displaymath}
\begin{array}{|c|c|c|c|c|c|c|c|}
\hline
  (i,j) =  &(6,1 )   &(6,2 )   &(6,3 )   &(6,4 )   &(6,5 )   &(6,6 )   &(6,7 )  \\
\hline
\lambda_{i,j} & 0.500000   & 0.500000   & 0.530000   & 0.550000   & 0.565000   & 0.565000   & 0.600000  \\
\hline
\mathcal{P}(f^{i,j}) & 1.586853   & 1.586853   & 1.582237   & 1.579160   & 1.576853   & 1.576853   & 1.571468  \\
\hline
\mathcal{P}(g^{i,j}) & 1.609785   & 1.598326   & 1.586682   & 1.595657   & 1.603117   & 1.605107   & 1.589894  \\
\hline
\mathcal{P}(f^{i,j}) \mathcal{P}(g^{i,j}) & 2.554493   & 2.536309   & 2.510507   & 2.519798   & 2.527881   & 2.531019   & 2.498468  \\
\hline
\end{array}
\end{displaymath}
\begin{displaymath}
\begin{array}{|c|c|c|c|c|c|c|c|}
\hline
  (i,j) =  &(7,1 )   &(7,2 )   &(7,3 )   &(7,4 )   &(7,5 )   &(7,6 )   &(7,7 )  \\
\hline
\lambda_{i,j} & 0.500000   & 0.515000   & 0.535000   & 0.555000   & 0.565000   & 0.570000   & 0.600000  \\
\hline
\mathcal{P}(f^{i,j}) & 1.568686   & 1.560602   & 1.549821   & 1.539044   & 1.533655   & 1.530958   & 1.517143  \\
\hline
\mathcal{P}(g^{i,j}) & 1.621572   & 1.609605   & 1.602462   & 1.612258   & 1.622822   & 1.638219   & 1.624359  \\
\hline
\mathcal{P}(f^{i,j}) \mathcal{P}(g^{i,j}) & 2.543738   & 2.511952   & 2.483529   & 2.481335   & 2.488849   & 2.508043   & 2.464386  \\
\hline
\end{array}
\end{displaymath}

\section{Model File for GLPK and Usage of Our Program 2DHSH+.c}
\begin{verbatim}

param I:=50;

param c{i in 1..I}>=0;

param w{i in 1..I};

var x{i in 1..I}, integer, >=0;

maximize  f: sum{i in 1..I} w[i]*x[i] + (1-sum{i in 1..I} c[i]*x[i]) * 38/37;

s.t.    x0: sum{i in 1..I} c[i]*x[i] <= 1;
        x1: sum{i in 1..7} x[i] <= 1;
        x7: sum{i in 8..13} x[i] <= 2;
        x14: x[14] <= 3;
        x15: x[15] <= 3;

        x16: x[16] <= 4;
        x17: x[17] <= 5;
        x18: x[18] + x[19] <= 6;
        y715: 2*x[7]  + x[15]  <= 3.9;
        y71317: 3*x[7] + 2*x[13] + x[17] <= 5.9;
        y131524: 4*x[13] + 3*x[15] + x[24] <= 11.9;
        y71118:  5*x[7] + 3.53*x[11] + 1.47 *x[18] <= 9;
        y7132036: 12*x[7]+8*x[13] + 3*x[20] + x[36] <=23;
        y7132130: 9*x[7] + 6*x[13] + 2*x[21] + x[30] <=17;
         others{i in 20..50}: x[i] <= i -13;
end;

 \end{verbatim}

%%\section{Usage of Our Program: 2DHSH+}
Whene the parameters in Super Harmonic such as $\alpha_i$, $\beta_i$,
$\gamma_i$ and $\phi(i)$ and $\varphi(i)$ are given,
we can calculate the weighting functions of Super Harmonic $W_B^{j}(\cdot)$.
Then the weighting functions $W^{i,j}(x,y)$ for algorithm $H \otimes SH+$
can be calculated as well as $f^{i,j}(y)$ and $g^{i,j}(x)$.
For each $(i,j)$, we call API of GLPK to solve $\mathcal{P}(f^{i,j})$ and  $\mathcal{P}(g^{i,j})$.

To use our program under {\em linux} system:
\begin{itemize}
\item Install GLPK,
\item Compile: ``gcc -o 2DHSH+.exe 2DHSH+.c -lglpk''
\item Run:  ``./2DHSH+.exe new2 $>$ yourfile''
\end{itemize}
If there is an error message like ``Could not load *.so'' 
 when you  compile the source,
 then try to set "LD\_LIRARY\_PATH" as follows:
  ``LD\_LIRARY\_PATH= \$LD\_LIRARY\_PATH:/usr/local/lib'',  then ``export LD\_LIRARY\_PATH''.

%  If an item of type-i has size in $(t_{i+1},t_i]$

%  Next we are going to explain why we use the above IP to calculate
%  the total weight in an one dimensional bin.
%  For example, let us see the constraint "y71118".

\end{document}